\documentclass[a4paper,11pt]{article}
\usepackage{pos}

\title{The  continuum limit of the modular discretization of AdS$_2$}

\author[a]{Minos Axenides}
\author*[b,a]{Emmanuel Floratos} 
\author[c]{Stam Nicolis}

\affiliation[a]{Institute of Nuclear and Particle Physics, NCSR ``Demokritos''\\
Aghia Paraskevi, GR--15310, Greece}

\affiliation[b]{Physics Department, University of Athens, Zografou University Campus\\
Athens, GR-15771, Greece}

\affiliation[c]{Institut Denis Poisson, Université de Tours, Université d'Orléans, CNRS (UMR7013)\\
Parc Grandmont, 37200 Tours, France}
\emailAdd{axenides@inp.demokritos.gr}
\emailAdd{mflorato@phys.uoa.gr}
\emailAdd{stam.nicolis@lmpt.univ-tours.fr}

\abstract{According to the ’t Hooft–Susskind holography, the black hole entropy, $S_\mathrm{BH},$ is
carried by the chaotic microscopic degrees of freedom, which live in the near horizon region
and have a Hilbert space of states of finite dimension $d = \exp(S_\mathrm{BH}).$ In previous work we
have proposed that the near horizon geometry, when the microscopic degrees of freedom can
be resolved, can be described by the AdS$_2[\mathbb{Z}_N ]$ discrete, finite and random geometry, where
$N\propto S_\mathrm{BH}.$  What had remained as an open problem is how the smooth
AdS$_2$ geometry can be recovered, in the limit when $N\to\infty.$ In this contribution,  we present the salient points of the solution to 
this problem, which involves  embedding  the discrete and finite AdS$_2[\mathbb{Z}_N ]$ geometry 
in a family of finite geometries, AdS$_2^M[\mathbb{Z}_N ],$ where $M$ is another integer. This family can be
constructed by an appropriate toroidal compactification and discretization of the ambient
(2+1)-dimensional Minkowski space-time. In this construction $N$ and $M$ can be understood
as “infrared” and “ultraviolet” cutoffs respectively. This construction allows us to
obtain the continuum limit of the AdS$_2^M[\mathbb{Z}_N]$ discrete and finite geometry, by taking both $N$
and $M$ to infinity in a specific correlated way, following a reverse process: Firstly, by  recovering the continuous, toroidally compactified, AdS$_2[\mathbb{Z}_N ]$ geometry, by removing the ultraviolet cutoff; secondly,  by removing the infrared
cutoff,  in a specific decompactification limit, while keeping the radius of AdS$_2$ finite. It is
in this way that we recover the standard non-compact AdS$_2$ continuum space-time. This
method can be applied directly to higher-dimensional AdS spacetimes. 
}

\FullConference{%
  Corfu Summer Institute 2021 "School and Workshops on Elementary Particle Physics and Gravity"\\
  29 August - 9 October 2021\\
  Corfu, Greece
}

 \tableofcontents

\usepackage{subfigure}
\usepackage{graphicx}

\usepackage{amsthm,bm}
\newtheorem{prop}{Proposition}

\newcommand{\newc}{\newcommand}
\newc{\be}{\begin{equation}}
\newc{\ee}{\end{equation}}

\begin{document}
\maketitle
\section{Introduction}\label{intro}
The present work, mathematically, belongs to the area of algebraic geometry over finite rings. However its relevance for physics  stems   from  the proposal of using specific,  
discrete and finite arithmetic geometries, as toy models,  in order to describe  properties of quantum gravity in general and the  structure of space-time, in particular,  at   distances of the order of the  Planck scale($10^{-33} \mathrm{cm}$), where the notions  of the  metric and  of the  continuity of spacetime break down~\cite{Axenides:2013iwa}. 

At Planck scale energies,    quantum mechanics, as we know it from lower  energy scales, implies that the notion of spacetime itself becomes ill-defined, through the appearance from the vacuum of real or virtual  black holes of Planck length size~\cite{Hawking:1979zw}.

  Probing this scale by scattering experiments of any sort of particle--like objects, black holes will be produced and the strength of the gravitational interaction will be of O(1), which leads to a breakdown of perturbative gravity and of  the usual continuum spacetime description~\cite{Carlip:2009km,Carlip:2011tt}. 

 The above remarks led some authors to consider the idea, that  one has to abandon continuity of spacetime, locality of interactions and regularity of dynamics.
  Indeed there are recent  arguments that quantization of gravity implies discretization and finiteness of space time~\cite{Hooft:2016pmw,Verlinde:2010hp}. This is, indeed,  an old idea, that was put forward, already way back, by the founders of quantum physics  and gravity.

A few years ago  the seminal paper~\cite{Almheiri:2012rt} highlighted the relevance of the so--called ``new black hole information paradox''~\cite{Papadodimas:2012aq}, 
 which finally lead to the conjectures that go under the label  ER=EPR~\cite{Maldacena:2013xja} and culminate in the  so--called QM=GR correspondence~\cite{Susskind:2017ney}.  

These conjectures relate strongly  the description of spacetime geometry and quantum gravity  to  quantum  information theoretic tools, such as entanglement of information, algorithmic complexity,random  quantum networks,quantum holography, error correcting codes. 

A discrete and finite spacetime for quantum gravity is a possible way for describing the remarkable fact that  the Hilbert space of states of the BH  microscopic degrees of freedom is finite-dimensional. Its dimensionality equals to the exponential of the Bekenstein-Christodoulou-Hawking  black hole  entropy, which is of quantum origin. The  generalization of the Bekenstein entropy bounds implies that,  for any pair of local observers in a general gravitational background, the physics inside  their causal diamond is also described by a finite dimensional  Hilbert space of states~\cite{Bousso:2018bli}. This result  has been exploited  further and consistently under the name of  holographic spacetime, in the works of refs.~\cite{Giddings:2012bm,Banks:2020dus,Bao:2017rnv}.

Our idea about the nature of spacetime at the  Planck scale, takes the  notion of a holographic spacetime  one step further:  Namely, that the finite dimensionality of the Hilbert space of local spacetime regions originates from  a  discrete and finite  spacetime, which underlies the emergent continuous geometric description~\cite{Floratos:1989au,Axenides:2013iwa}. 

Our starting point, therefore, is  the  hypothesis that  space-time, at the Planck scale, is  fundamentally  discrete and finite
and, moreover, does not emerge from any other continuous description(conformal field theory,  string theory, or anything else). We claim  that, at ``large'' distances (in units of the Planck length), the continuous spacetime geometry can be described as an  infrared limit thereof. This hypothesis, indeed, is  similar to the proposal by 't Hooft~\cite{Hooft:2016pmw}. 

This assumption implies  developing and using the appropriate mathematical tools, that can describe the properties and dynamics of discrete and finite geometries as well as the emergence, in their infrared limits,  of continuous geometries. So what we shall show in this contribution is an explicit example of how the continuous geometry of AdS$_2$ can emerge as a scaling limit of a specific discretization procedure. The crucial insight is that, in order to obtain a scaling limit, the  discretization procedure must entail the  introduction of {\em two} cutoffs, a UV cutoff and an IR cutoff, related in a particular way-a more complete presentation may be found in ref.~\cite{Axenides:2019lea}.

We do not wish to imply that it is not possible to define quantum gravity, with a finite dimensional Hilbert space, in any other way; just that this is one possible way to describe quantum physics with finite dimensional Hilbert space.

\section{Discretization and toroidal compactification of the AdS$_2$ geometry}\label{modN}

\subsection{The UV cutoff,  the lattice of integral points and the SO$(2,1,\mathbb{Z})$ isometry of AdS$_2^M[\mathbb{Z}]$ }\label{UVcutofflatt}
 We shall now present and study in detail  the lattice of  integral points of AdS$_2,$ along with its isometries. 
 
The physical lengthscale in our problem is the radius of the AdS$_2$ spacetime, $R_{\mathrm{AdS}_2}.$ We set   $R_{\mathrm{AdS}_2}=1$ and we divide it into $M$ segments, of length $a=R_{\mathrm{AdS}_2}/M.$
This defines  $a$ as the UV cutoff (lattice spacing) and $M\in\mathbb{N}$ and, hence, a  lattice in $\mathscr{M}^{2,1}.$ 

The continuum limit is defined by taking  $M\to\infty$ and $a\to 0$ with $R_{\mathrm{AdS}_2}=1$ fixed. 

The global embedding coordinates $(x_0, x_1,x_2)$ of this lattice are $(ka,la,ma)=a(k,l,m),$ where $k,l,m\in\mathbb{Z}.$ They  are measured in units of the lattice 
spacing $a.$ Therefore the lattice points, that lie on AdS$_2$ satisfy the equation 
\begin{equation}
\label{AdS2latt}
k^2+l^2-m^2=M^2
\end{equation}
whose solutions define AdS$_2^M[\mathbb{Z}],$ the set of all integral points of AdS$_2$ with integer radius $M.$ 
 
 In the literature there has been considerable effort in counting the number of solutions to the above equation, in particular the asymptotics of the density of such points~\cite{Shanks,Baragar,Kontorovich,Duke,lowryduda2017variants,oh2014limits}. 
This problem can be mapped  to a problem whose solution is known, namely the Gauss circle problem. This pertains to  finding the number $r_2(m,M)$  of solutions to the equation $k^2+l^2=M^2+m^2.$  This number is determined by factoring $M^2+m^2$ into its prime factors~\cite{Shanks} and counting the
number of primes, $p_i,$ of the form $p_i\equiv1\,\mathrm{mod}\,4$~(this is described in detail in~\cite{bressoudwagon};  the dependence on $M$ is a topic of current research~\cite{lowryduda2017variants,oh2014limits}). 

This factorization procedure generates a sequence of primes that contains an element of inherent randomness. It is this property that captures the random distribution of the integral points on AdS$_2$--this is illustrated in figs.~\ref{AdS2rat}.
\begin{figure}[thp]
\begin{center}
\subfigure{\includegraphics[scale=0.5]{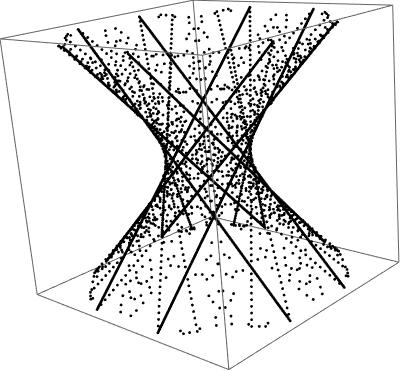}}
\subfigure{\includegraphics[scale=0.5]{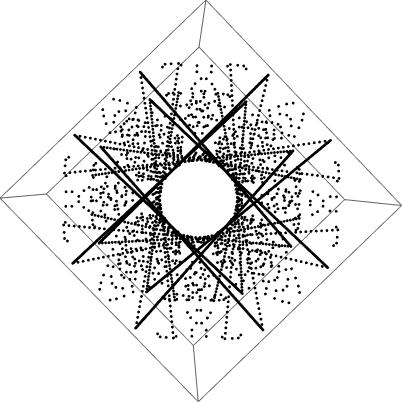}}
\end{center}
\caption[]{Integral points on AdS$_2.$}
\label{AdS2rat}
\end{figure}

Therefore, from these facts,  the number of integral points of the hyperboloid, up to height $m,$  is given by the expression
\be
\label{solhyp}
\mathrm{Sol}(m)=4+2\sum_{j=1}^m r_2(j,M)
\ee
We plot this function--in fig.~\ref{GPpoints}, for $M=1$, when $m$ runs from $-200$ to 200 (due to the  symmetry, $m\leftrightarrow -m,$ we plot only the positive values of $m.$)
\begin{figure}[thp]
\begin{center}
\includegraphics[scale=0.3]{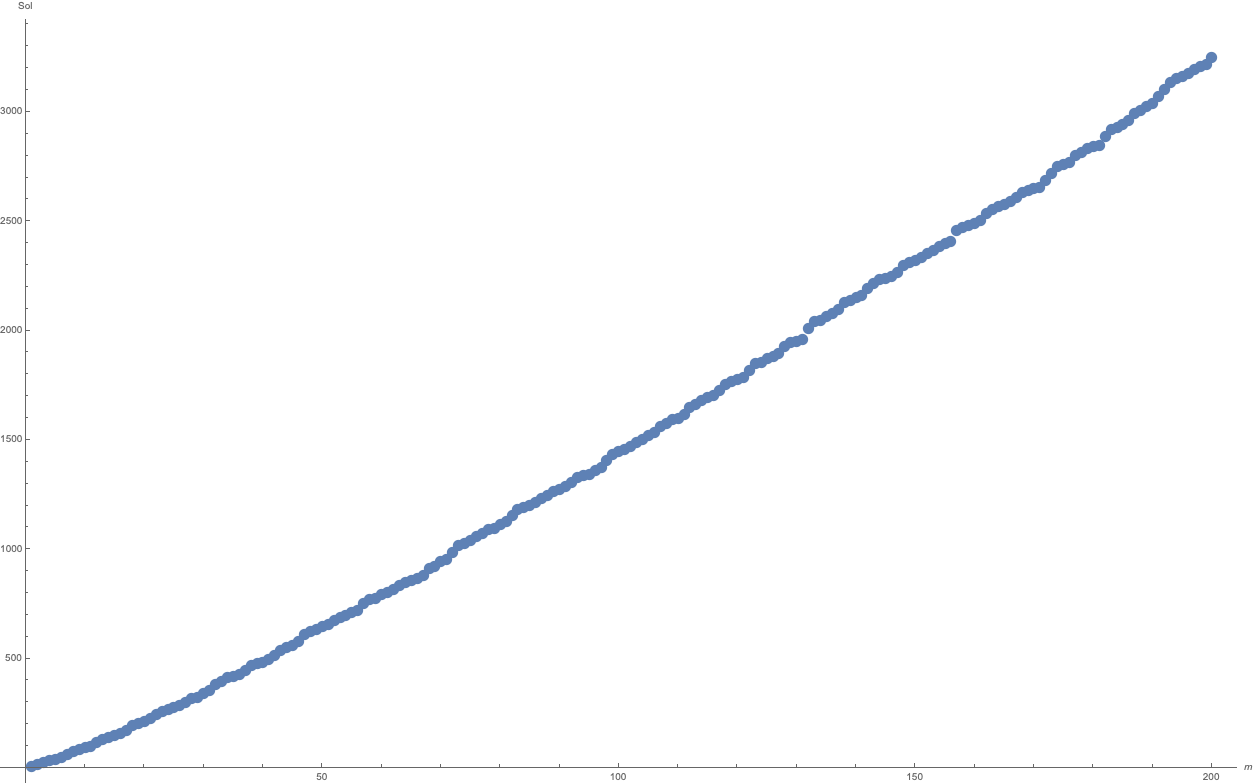}
\end{center}
\caption[]{The number of integral points, on AdS$_2,$ as a function of the height, $m,$ for $M=1$. Due to symmetry, $m\leftrightarrow -m,$ we plot only the positive values of $m.$}
\label{GPpoints}
\end{figure}

It is, indeed, striking that the result is an almost  straight line ~\cite{lowryduda2017variants,oh2014limits}.

We shall now discuss how to actually construct these points, using the property  that they belong to light--cone lines,  which emerge from the rational points of the circle on the throat of AdS$_2.$

Using the ruling property of AdS$_2,$
\be
\label{coords}
\begin{array}{l}
k = \cos\phi - \mu\sin\phi\\
l = \sin\phi  + \mu\cos\phi\\
m = \mu
\end{array}
\ee
we may repackage these as follows
\be
\label{coords1}
x_0 + \mathrm{i}x_1 = k+\mathrm{i}l=e^{\mathrm{i}\phi}(1+\mathrm{i}\mu)=e^{\mathrm{i}\phi}(1+\mathrm{i}m)\Leftrightarrow e^{\mathrm{i}\phi}=\frac{k+\mathrm{i}l}{1+\mathrm{i}m}
\ee
hence
\be
\label{rationalpointscircle}
\begin{array}{lll}
\displaystyle
\cos\phi = \frac{k+lm}{1+m^2} & \mathrm{and}& \displaystyle\sin\phi = \frac{l-mk}{1+m^2}
\end{array}
\ee
We remark that these are rational numbers--therefore they label rational points on the circle~\cite{tan}. 

The light cone lines at $(k,l,m)$ are, therefore,  parametrized by $\mu\in(-\infty,\infty)$, as
\be
\label{lcklm}
\begin{array}{l}
x_0 =  \frac{k+lm}{1+m^2} -\mu\frac{l-mk}{1+m^2}\\
x_1 = \frac{l-mk}{1+m^2} +\mu\frac{k+lm}{1+m^2}\\
x_2 = \mu
\end{array}
\ee
(When $\mu=x_2=m,$ $x_0=k$ and $x_1 = l.$)
\begin{prop}
On these  specific light-cone lines there exist infinitely many integral points,
 when $\mu,$ that labels the space--like direction $x_2,$  takes appropriate integer values.
\end{prop}
\begin{proof}
We write 
\be
\label{intpt1}
x_0(\mu)+\mathrm{i}x_1(\mu)=e^{\mathrm{i}\phi}(1+\mathrm{i}\mu)
\ee
where $\phi$ is defined by eq.~(\ref{rationalpointscircle}).

We look for integer values of $\mu=n\in\mathbb{Z},$ such  that $x_0(n)$ and $x_1(n)$ are, also, integers. 

That is  
\be
\label{gaussianInt}
x_0(n)+\mathrm{i}x_1(n)=\frac{k+\mathrm{i}l}{1+\mathrm{i}m}(1+\mathrm{i}n)
\ee
should be a Gaussian integer and this can hapen iff  $(1+\mathrm{i}n)/(1+\mathrm{i}m)=a+\mathrm{i}b$ with $a,b\in\mathbb{Z}.$ 

Therefore
\be
\label{gaussianint1}
1+\mathrm{i}n=(a-mb)+\mathrm{i}(am+b)\Leftrightarrow\left\{\begin{array}{l} 1 = a-mb\\ n = am+b\end{array}\right.
\ee

Thus on the light cone line passing through the point $(k,l,m)$ there are infinite integer points parametrized as:
\be
\label{pointonAdS2}
\begin{array}{l}
\displaystyle
x_0 = k + b(km-l)\\
\displaystyle
x_1 = l + b(k+lm)\\
\displaystyle
x_2 = n = m+ b(1+m^2)
\end{array}
\ee
\end{proof} 
\begin{prop}
Conversely, on  any light cone line emanating from any rational point of the circle on the throat of the hyperboloid there is an infinite number of integer points. 
\end{prop}
\begin{proof}
Indeed,we have
\be
\label{gaussianint2}
e^{\mathrm{i}\phi}\equiv\frac{a+\mathrm{i}b}{a-\mathrm{i}b}\Leftrightarrow x_0+\mathrm{i}x_1=\frac{a+\mathrm{i}b}{a-\mathrm{i}b}(1+\mathrm{i}n)
\ee
with $a,b\in\mathbb{Z}.$ In order to to obtain an integral point, for $\mu=n,$ we must have 
\be
\label{gaussianint3}
\frac{1+\mathrm{i}n}{a-\mathrm{i}b}=d+\mathrm{i}c
\ee
with $c,d\in\mathbb{Z}$ 

We immediately deduce that 
\be
\label{gaussianint4}
\begin{array}{l}
\displaystyle
1 = ad-bc\\
\displaystyle
n = ac + bd\\
\end{array}
\ee
These expressions imply that, given the integers $a$ and $b,$ it's possible to find the integers $c$ and $d$ and to express the coordinates $x_0, x_1$ and $x_2$ as
\be
\label{gaussianint5}
\begin{array}{l}
\displaystyle 
x_0 = ad+bc\\
\displaystyle
x_1 = ac-bd\\
\displaystyle
x_2 = ac+bd
\end{array}
\ee
The Diophantine equation $1 = ad-bc$ is solved for $c$ and $d,$ given two coprime integers $a$ and $b,$ by the Euclidian algorithm--which seems to  lead to a unique solution, implying that the point $(x_0, x_1, x_2)$ is unique. 

However there's a subtlety! There are {\em infinitely many} solutions $(c,d),$ to the equation $ad-bc=1$! The reason is that, given any one solution $(c,d),$ the pair 
$(c+\kappa a, d+\kappa b),$ with $\kappa\in\mathbb{Z},$ is, also, a solution, as it can be checked by substitution.

Therefore there is a one--parameter family of points, labeled by the integer $\kappa$:
\be
\label{integerpointAdS2}
\begin{array}{l}
\displaystyle
x_0 = ad+bc + 2\kappa a b\\
\displaystyle
x_1 = ac - bd + \kappa (a^2-b^2)\\
\displaystyle
x_2 = ac+bd + \kappa (a^2+b^2)
\end{array}
\ee
We remark, however,  that the vector $(2ab,a^2-b^2,a^2+b^2)$ is light--like, with respect to the $(++-)$ metric: $(2ab)^2+(a^2-b^2)^2-(a^2+b^2)^2=0.$ So eq.~(\ref{integerpointAdS2}) 
describes a shift of the  point $(ad+bc, ac-bd, ac+bd),$ along a light--like direction. Since the shift is linear in the ``affine parameter'', $\kappa,$ it generates a light--like line, passing through the original point.  

In this way we have established the dictionary between the rational points of the circle and the integral points of the hyperboloid. 

\end{proof}

Now we proceed with the study of the discrete symmetries of  the integral Lorentzian  lattice of  $\mathscr{M}^{2,1},$ where the lattice of integral points on AdS$_2$ is embedded.
The lattice of integral points of $\mathscr{M}^{2,1},$ with one space-like and two time-like dimensions, carries as isometry group the group of integral Lorentz boosts SO$(2,1,\mathbb{Z}),$ as well as integral Poincaré translations. The double cover of this infinite and discrete group is SL$(2,\mathbb{Z}),$ the modular group. This has been shown by Schild~\cite{Schild49,Schild48} in the 1940s.  The group SO$(2,1,\mathbb{Z})$ can be generated by reflections, as has been shown by Coxeter~\cite{Coxeter},  Vinberg~\cite{Vinberg_1967}.  This work culminates in the famous book by Kac~\cite{Kac1990}, where he   introduced the notion of hyperbolic, infinite dimensional, Lie algebras. The characteristic property of  such algebras is that the discrete Weyl group of their root space is an integral Lorentz group. Generalization from $SL(2,\mathbb{Z})$ to other normed algebras has been studied in~\cite{Feingold:2008ih}. 

The fundamental domain of SO$(2,1,\mathbb{Z})$  is the minimum
 set of points  of the integral lattice of $\mathscr{M}^{2,1},$which are
not  related by any element of the group  and from which,
  all the other points of the lattice can be generated by repeated action
of  the elements of the group. It turns out that the fundamental region
is an infinite set of points which  can be generated by repeated action of  reflections in the following way:

Using the metric $h\equiv\mathrm{diag}(1,1,-1)$ on $\mathscr{M}^{2,1}$  the generating reflections, elements of SO$(2,1,\mathbb{Z}),$ are given by the matrices
\begin{equation}
\label{genrefl}
\begin{array}{lclcl}
R_1=\left(\begin{array}{ccc} 
-1 & & \\ 
    & 1 & \\
    &   & 1
\end{array}\right),  & & 
R_2 = \left(\begin{array}{ccc} 
1 & & \\
   & 1 & \\
   &   & -1
\end{array}\right),  & & 
R_3 = \left(\begin{array}{ccc} 
0 & 1 & \\
1 & 0 & \\
   &   & 1
\end{array}\right) \\
& & 
R_4 =   \left(\begin{array}{ccc} 
1 & -2 & -2\\
2 & -1 & -2\\
-2    &  2 & 3
\end{array}\right)
\end{array}
\end{equation} 
If $(k,l,m)$ are the coordinates of the integral lattice, the fundamental domain of SO$(2,1,\mathbb{Z})$ can be defined by the conditions $m\geq k+l\geq 0$ and $k\geq l\geq 0.$ This fundamental domain, restricted on AdS$_2^M[\mathbb{Z}],$ defines the corresponding fundamental domain of SO$(2,1,\mathbb{Z}),$ acting on AdS$_2^M[\mathbb{Z}].$ 
This region of AdS$_2[\mathbb{Z}]$ lies in the positive octant of $\mathscr{M}^{2,1}$ and between the two planes, that define the conditions--cf.fig.~\ref{FundDomainAdS2ZN}. It is of {\em infinite} extent.
\begin{figure}[thp]
\begin{center}
\includegraphics[scale=0.6]{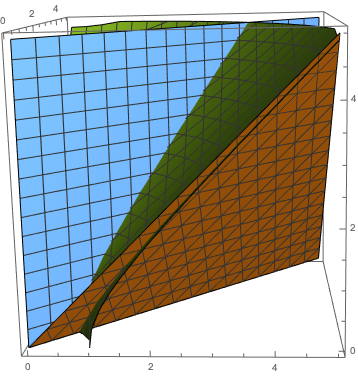}
\end{center}
\caption[]{The fundamental domain of SO$(2,1,\mathbb{Z})$ on AdS$_2^M[\mathbb{Z}]$ is the dark green  part of the hyperoboloid, in the positive octant, that lies between the two planes, 
$m\geq k+l\geq 0$ and $k\geq l\geq 0.$}
\label{FundDomainAdS2ZN}
\end{figure}

\subsection{The IR cutoff and the  toroidal compactification of AdS$_2$}\label{UVIR}
Having introduced the lattice of integral points on AdS$_2,$ which we consider as defining an UV cutoff, we proceed, now, to impose an infrared (IR) cutoff. The crucial reason for such a cutoff is that in order to study chaotic   Hamiltonian  dynamics on this spacetime.~\cite{Axenides:2016nmf}, we have made use of the interpretation of AdS$_2$ as a phase space of  single particles, due to the symplectic nature of the isometry $SL(2,\mathbb{R})=Sp(2,\mathbb{R})$. The additional requirement of  mixing (scrambling) imposes the condition of the compactness of the phase space and therefore the necessity of imposing  of an infrared cutoff (for a detailed discussion of this point cf.~\cite{ArnoldAvez}). 

Having  embedded the AdS$_2$ hyperboloid, 
\begin{equation}
\label{AdS2_M21}
x_0^2 + x_1^2 - x_2^2 = R_{\mathrm{AdS}_2}^2
\end{equation}
in $\mathscr{M}^{2,1},$ the IR  cutoff, $L$  is defined by periodically identifying all the spacetime points of $\mathscr{M}^{2,1},$ if the difference of their coordinates is an integral vector$\times L$: 
\begin{equation}
\label{equivclassIR}
x\sim y\Leftrightarrow x-y = (k,l,m)L
\end{equation}
where $k,l,m\in\mathbb{Z}.$ 
In this way we have compactified $\mathscr{M}^{2,1}$ to the three-dimensional torus, of size $L,$ $\mathbb{T}^3(L).$ 

More concretely, $\mathbb{T}^3(L)$ is the fundamental domain of the group of integral translations, $\mathbb{Z}\times\mathbb{Z}\times\mathbb{Z},$ acting on $\mathscr{M}^{2,1}.$ To describe this geometric property by the algebraic operation, mod $L,$ that acts on the coordinates of $\mathscr{M}^{2,1},$ we are led to identify the fundamental domain with the positive octant of $\mathscr{M}^{2,1},$ i.e. $x_0,x_1,x_2\geq 0$

After this compactification, the spacetime geometry of AdS$_2$ becomes a foliation of the 3-torus, with leaves the images of AdS$_2$ under the operation mod $L.$ So the equation, whose solutions define the points of the compactified AdS$_2,$ is 
\begin{equation}
\label{AdS2comp}
x_0^2+x_1^2-x_2^2\equiv\,R_\mathrm{AdS_2}^2\,\mathrm{mod}\,L
\end{equation}
where $(x_0,x_1,x_2)\in\mathbb{T}^3(L).$

It is obvious, that inside the 3-torus, there is a part of the AdS$_2$ surface, which corresponds to solutions of eq.~(\ref{AdS2comp}), without the mod $L$ operation. On the other hand, the infinite part of AdS$_2,$ that lies outside the torus, is partitioned in infinitely many pieces, which belong to images of $\mathbb{T}^3(L)$ in $\mathscr{M}^{2,1}.$ 
 These pieces  are brought inside the  torus by the mod $L$ operation.

Now we choose the IR cutoff $L$ in units of $a,$ so that $L=aN,$ where $N$ is an integer,  independent of $M.$ It is constrained by $N>M,$ since the cube should contain, at least, the throat of AdS$_2.$ 

So the scaling limit entails  taking  $M\to\infty,$ $N\to\infty,$ but keeping $L$ fixed. 

The periodic nature of the IR cutoff implies that we must take the images of all integral points of AdS$_2[\mathbb{Z}]$ under the mod $N$ operation, inside the cubic lattice of $N^3$ points. 

The set of these images satisfy the equations
\begin{equation}
\label{AdS2ZN}
k^2+l^2-m^2\equiv M^2\,\mathrm{mod}\,N
\end{equation} 
The set of points satisfying this condition will be called AdS$_2^M[\mathbb{Z}_N]$. 

 Our definition for AdS$_2[\mathbb{Z}_N]$ in our previous work  was similar to the one given here. The only difference being that the RHS of eq.~(\ref{AdS2ZN}) was 
$1\,\mathrm{mod}\,N,$ which was chosen for convenience, rather than for any intrinsic reason. We remark that the two definitions are consistent iff
$M^2\equiv\,1\,\mathrm{mod}\,N.$ 

The solutions of eq.~(\ref{AdS2ZN}), when $M^2\equiv 1\,\mathrm{mod}\,N,$ produce the  AdS$_2[\mathbb{Z}_N]$ geometry introduced in our previous work. 

\section{Continuum limit for large $N$}\label{contlim}
\subsection{Constraints on the double sequences of the  UV/IR cutoffs}\label{M2N}

Having constructed the finite geometry, AdS$_2^M[\mathbb{Z}_N]$ and established its relation with AdS$_2[\mathbb{Z}_N]$, we shall discuss the meaning of the limit, $M,N\to\infty$. It is in this limit that we hope to recover the continuum AdS$_2$ geometry. 

Such a  limit  can be defined using the topology of the  ambient  Minkowski spacetime $\mathscr{M}^{2,1}$.

Specifically, we use a  reverse, two--step, process: Firstly, by removing the UV cutoff; next, by removing  the IR cutoff. This is realized  by choosing any sequence of pairs of integers, $(M_n,N_n),$ $n=1,2,3,\ldots,$ such that, for any $n=1,2,3,\ldots$ 
\begin{itemize}
\item
 $N_n>M_n,$
\item  $M_n^2\equiv\,1\,\mathrm{mod}\,N_n,$ 
\item The limit of the ratio $N_n/M_n$ takes a finite value, $>1$ (as $n\to\infty$), which we can  identify with $L/R_{\mathrm{AdS}_2}.$  
\end{itemize}
Below we shall present the general solution to the equation $M^2\equiv\,1\,\mathrm{mod}\,N.$ Subsequently, we shall select those solutions that satisfy the other requirements. 

The first step is to factor $N$ into (powers of) primes,  $N=N_1\times N_2\times\cdots\times N_l=q_1^{k_1}q_2^{k_2}\cdots q_l^{k_l}.$ Then the equation $M^2\equiv\,1\,\mathrm{mod}\,N,$ is equivalent to the system
\begin{equation}
\label{M21N}
M_I^2\equiv\,1\,\mathrm{mod}\,q_I^{k_I}
\end{equation} 
where $I=1,2,\ldots,l.$ The Chinese Remainder Theorem~\cite{bressoudwagon}  then implies that all the solutions of eq.~(\ref{M21N}) can be used to construct $M,$ with $M=M_1m_1n_1 + \cdots M_lm_ln_l,$ where 
$M_I\equiv\,M\,\mathrm{mod}\,N_I,$ $m_I=N/N_I,$ $n_I\, \equiv m_I^{-1}\,\mathrm{mod}\,N_I.$

When $q_I\neq 2,$ the solutions are  $M_I=1$ and $q_I^{n_I}-1.$ When $q_I = 2,$ there exist four solutions, $M_I=1,2^{n_I}-1,2^{n_I-1}\pm 1.$ 

Now we must choose sequences, $N_n$ and determine the corresponding $M_n,$ satisfying the constraints listed above. 

In the next  two subsections  we shall present nontrivial examples of sequences of pairs, $(M_n, N_n)$ satisfying the above constraints, whose limiting ratio, $\lim_{n\to\infty} N_n/M_n,$ is the ``golden'' or ``silver'' ratios. 
The general question of determining sequences which have an arbitrary, but given, limiting ratio, is an interesting question, which is deferred to a future work.

\subsection{Removing the UV cutoff by the Fibonacci sequence}\label{fibon}
Although it is easy to demonstrate the existence of such sequences--for example, $N_n=2^n$ and $M_n=2^{n-1}\pm 1,$ where $M_n^2\equiv\,1\,\mathrm{mod}\,N_n$ and $N_n/M_n\to 2,$  which implies that $L/R_{\mathrm{AdS}_2}=2$, in this section we focus on another particular class of sequences, based on the Fibonacci integers, $f_n$~\cite{bressoudwagon}.  This case is of particular interest, since, in our previous paper~\cite{Axenides:2016nmf}, where we studied fast scrambling, we found that, for geodesic observers, moving in AdS$_2[N],$ with evolution operator the Arnol'd cat map, the fast scrambling bound is saturated, when $N$ is a Fibonacci integer. 

The Fibonacci sequence, defined by 
\begin{equation}
\label{fibinacci_seq}
\begin{array}{l}
f_0=0; f_1=1\\
f_{n+1}=f_n+f_{n-1}\\
\end{array}
\end{equation}
can be written  in matrix form
\begin{equation}
\label{matrixfib}
\left(\begin{array}{c} f_n\\f_{n+1}\end{array}\right)=\underbrace{\left(\begin{array}{cc} 0 & 1 \\ 1 & 1\end{array}\right)}_{\sf A}\left(\begin{array}{c} f_{n-1} \\ f_n\end{array}\right)
\end{equation}
We remark that the famous Arnol'd cat map can be written as 
\be
\label{ArnoldCM}
\left(\begin{array}{cc} 1 & 1 \\ 1 & 2\end{array}\right) = {\sf A}^2
\ee
Since the matrix ${\sf A}$ doesn't depend on $n,$ we can solve the recursion relation in closed form, by setting $f_n\equiv C \rho^n$ and find the equation, satisfied by $\rho$
$$
\rho^{n+1}=\rho^n+\rho^{n-1}\Leftrightarrow \rho^2-\rho-1=0\Leftrightarrow \rho\equiv \rho_\pm=\frac{1\pm\sqrt{5}}{2}
$$
Therefore, we may express $f_n$ as a linear combination of $\rho_+^n$ and $\rho_-^n=(-)^n\rho_+^{-n}$:
\be
\label{solfib}
f_n=A_+\rho_+^n+A_-\rho_-^n\Leftrightarrow\left\{\begin{array}{l} f_0=A_+ + A_- = 0\\ f_1 = A_+\rho_+ + A_-\rho_-=1\end{array}\right.
\ee
whence we find that 
$$
A_+ = -A_-=\frac{1}{\rho_+ -\rho_-}=\frac{1}{\sqrt{5}}
$$
therefore,
\be
\label{solfib1}
f_n=\frac{\rho_+^n-(-)^n\rho_+^{-n}}{\sqrt{5}}
\ee
It's quite fascinating that the LHS of this expression is an integer!

The eigenvalue $\rho_+>1$ is known as the ``golden ratio'' (often denoted by $\phi$ in the literature) and it's straightforward  to show that $f_{n+1}/f_n\to\rho_+,$ as $n\to\infty.$

Furthermore, it can be shown, by induction,  that the elements of ${\sf A}^n$ are, in fact, the Fibonacci numbers themselves, arranged as follows:
 \begin{equation}
 \label{fibonacciA}
 {\sf A}^n = \left(\begin{array}{cc} f_{n-1} & f_n\\ f_n &  f_{n+1}\end{array}\right)
 \end{equation}
 One reason this expression is useful is that it implies that $\mathrm{det}\,{\sf A}^n = (-)^n=f_{n-1}f_{n+1}-f_n^2.$ 
 
 For $n=2l+1,$ we remark that this relation takes the form  $f_{2l+1}^2=1+f_{2l}f_{2l+2}.$ 
 
 Now, since $f_{2l+1}$ and $f_{2l+2}$ are successive iterates, they're coprime, which  implies, that $f_{2l+1}^2\equiv\,1\,\mathrm{mod}\,f_{2l+2}.$
 
Therefore, the sequence of pairs,
$(M_l=f_{2l+1},N_l=f_{2l+2}),$ where $l=1,2,3,\ldots,$ satisfy all of the requirements and the
corresponding limiting ratio, $L/R_\mathrm{AdS_2},$ can be found analytically. It is, indeed, equal to $\rho_+=(1+\sqrt{5})/2,$ the golden ratio. 

 In the next subsection we shall consider the so-called $k-$Fibonacci sequences, which will be important for obtaining other values for the ratio $L/R_{\mathrm{AdS}_2}$, as well as for  removing the IR cutoff. 

\subsection{Removing the IR cutoff using the generalized $k-$Fibonacci sequences}\label{contfrac}
It's possible to generalize the Fibonacci sequence in the  following way:
\begin{equation}
\label{kfibrec}
g_{n+1}=kg_n+g_{n-1}
\end{equation}
with $g_0=0$ and $g_1=1$ and $k$ an integer.  This is known as the ``$k-$Fibonacci'' sequence~\cite{Horadam}. 

We may solve for $g_n\equiv C\rho^n$; the characteristic equation for $\rho,$ now, reads
\be
\label{kfibseq}
\rho^2-k\rho-1=0\Leftrightarrow\rho_\pm(k)=\frac{k\pm\sqrt{k^2+4}}{2}
\ee
and express $g_n$ as a linear combination of the $\rho_\pm$:
\be 
\label{kfibsol}
g_n = A_+\rho_+(k)^n + A_-\rho_-(k)^n=\frac{\rho_+(k)^n-(-)^n\rho_+(k)^{-n}}{\sqrt{k^2+4}}
\ee
that generalizes eq.~(\ref{solfib1}).

In matrix form
\be
\label{kfibinacci}
\left(\begin{array}{c} g_n\\ g_{n+1}\end{array}\right)=\underbrace{\left(\begin{array}{cc} 0 & 1\\ 1 & k\end{array}\right)}_{{\sf A}(k)}\left(\begin{array}{c} g_{n-1}\\g_n\end{array}\right)
\ee
Similarly as for the usual Fibonacci sequence, we may show, by induction, that 
\be
\label{Akn}
{\sf A}(k)^n=\left(\begin{array}{cc} g_{n-1} & g_n \\ g_n & g_{n+1}\end{array}\right)
\ee
We find that $\mathrm{det}\,{\sf A}(k)^n=(-)^n,$ therefore that $g_{2l+1}^2\equiv\,1\,\mathrm{mod}\,g_{2l+2}$; thus, $g_{2l+2}/g_{2l+1}\to L/R_\mathrm{AdS_2}=\rho_+(k),$ where the eigenvalue of ${\sf A}(k),$ $\rho_+(k),$ that's greater than 1, of course, depends on $k.$ In this way it is possible to obtain infinitely many values of the ratio $L/R_{\mathrm{AdS}_2}$. 
Furthermore,  we have determined $L,$ the IR cutoff,  in terms of $R_\mathrm{AdS_2}.$

What is remarkable is that, using the 
additional parameter, $k,$ of the $k-$Fibonacci sequence, it is, now,  possible to remove the IR cutoff, as well, since it is possible to send $L\to\infty,$ as $k\to\infty,$ keeping $R_\mathrm{AdS_2}$ fixed.  

While  $k$ remains finite, the periodic box cannot be removed and, in the continuum limit, $a\to 0,$ we obtain infinitely many foldings of the AdS$_2$ surface inside the box due to the mod $L$ operation. 

The Fibonacci sequence, taken mod $N,$ is periodic, with period $T(N)$; this turns out to be a ``random'' function of $N.$ The ``shortest'' periods, as has been shown by Falk and Dyson~\cite{falk_dyson}, occur when $N=F_l,$ for any $l.$ In that case, $T(F_l)=2l.$ 

We may, thus, ask the same question for the $k-$Fibonacci sequence, where the ratio of its successive elements, $g_{n+1}/g_n$ tend to the so-called ``$k-$silver ratio'', 
\be
\label{silverratio}
\rho_+(k)=\frac{k+\sqrt{k^2+4}}{2}
\ee 
(the ``silver ratio'' is $\rho_+(k=2)$)

From eq.~(\ref{Akn}), taking mod $g_l$ on both sides, we find that, when $n=l,$ the matrix becomes $\pm$(the identity matrix), so $T(g_l)=l$ or $2l,$ respectively; thereby generalizing the Falk--Dyson result for the $k-$Fibonacci sequences. 
  
\section{Conclusions}\label{concl}
The logical approach for discussing the relation between classical and quantum physics involves showing how the former can be obtained as a limit of the latter, since it is the quantum description that is more ``fundamental'' and it isn't possible to describe quantum effects in terms of classical physics. In this contribution we have, therefore, shown how the smooth AdS$_2$ geometry, that is a hallmark of the near horizon geometry of extremal black holes, can, indeed, be obtained by a limiting process from a finite and discrete geometry, that has been shown to capture the consistent description of the single-particle probes of the near horizon geometry, that can resolve the individual black hole microstates. 

This approach can be readily generalized to higher dimensional AdS$_k$ spacetimes, $k>2.$ 

The next step involves describing the near horizon degrees themselves, as a many-body system. 
 
\newpage
\bibliographystyle{JHEP}
\bibliography{ads2discrete}
\end{document}